\documentclass{article}

\usepackage{arxiv}

\usepackage{amsthm,amsmath,latexsym,bbm,enumitem,xcolor}
\usepackage{amssymb}

\usepackage{graphicx}
\usepackage{mdframed}

\newtheorem{theorem}{Theorem}
\newtheorem{lemma}{Lemma}
\newtheorem{definition}{Definition}
\newtheorem{fact}{Fact}

\newtheorem{corollary}{Corollary}

\newcommand{\cF}{\mathcal{F}}

\newcommand{\remove}[1]{}

\title{Optimal Packet-oblivious Stable Routing in Multi-hop Wireless Networks}

\author{
  Vicent  Cholvi \\
  Departament de Llenguatges i Sistemes Inform\`{a}tics\\
  Universitat Jaume I\\
   Castell\'{o}, Spain \\
   \and
 Pawe{\l} Garncarek \\
  Instytut Informatyki\\
  Uniwersytet Wroc{\l}awski\\
   Wroc{\l}aw, Poland \\
   \and
   Tomasz Jurdzi\'{n}ski\\
  Instytut Informatyki\\
  Uniwersytet Wroc{\l}awski\\
   Wroc{\l}aw, Poland \\
   \and
   Dariusz R. Kowalski\\
   School of Computer and Cyber Sciences\\
   Augusta University\\
   Augusta, Georgia, USA
  }

\begin{document}

\maketitle

\begin{abstract}
Stability is an important issue in order to characterize the performance of a network, and it has become a major topic of study in the last decade. Roughly speaking, a communication network system is said to be \emph{stable} if the number of packets waiting to be delivered (backlog) is finitely bounded at any one time. 

In this paper, we introduce  a new family of combinatorial structures, which we call \emph{universally strong selectors}, that are used to provide a set of transmission schedules. Making use of these structures, combined with some known queuing policies, we  propose a packet-oblivious routing algorithm which is working without using any global topological information, and guarantees stability for certain injection rates.  We show that this protocol is asymptotically optimal regarding the injection rate for which stability is guaranteed.

Furthermore, we also introduce a packet-oblivious routing algorithm that guarantees stability for higher traffic. This algorithm is optimal regarding the injection rate for which stability is guaranteed. However, it needs to use some global information of the system topology. 
\end{abstract}


\section{Introduction}
\label{sec:introduction}

Stability is an important issue in order to characterize the performance of a network, and it has become a major topic of study in the last decade. Roughly speaking, a communication network system is said to be {\em stable} if the number of packets waiting to be delivered (backlog) is finitely bounded at any one time. The importance of such an issue is obvious, since if one cannot guarantee stability, then one cannot hope to be able to ensure deterministic guarantees for most of the network performance metrics.

For many years, the common belief was that only overloaded queues (i.e., when the total arrival rate is greater than the service rate) could generate instability, while underloaded ones could only induce delays that are longer than desired, but always remain stable. However, this belief was shown to be wrong when it was observed that, in some networks, the backlogs in specific queues could grow indefinitely even when such queues were not overloaded~\cite{BorodinKRSW01,AndrewsAFLLK01}. These later results aroused an interest in understanding the stability properties of packet-switched networks, so that a substantial effort has been invested in that area. Among the obtained results, stability of specific scheduling policies was considered for example in~\cite{BhattacharjeeGL-SICOMP04,CholviE07,Gamarnik-SICOMP03,Goel01}.  The impact of network topologies on injection rates that guarantee stability was considered in~\cite{EchagueCF03,KoukopoulosMNS-TCSy05,LotkerPR04}. A systematic account of issues related to universal stability was given in~\cite{AlvarezBS-SICOMP04}.

Whereas in  wireline networks a node can transmit data over any outgoing link and simultaneously receive data over any incoming link, the situation is different in wireless networks. Indeed, nearby wireless signal transmissions that overlap in time can interfere with one another, to the effect that none can be  transmitted successfully. This makes the study of stability in wireless networks more complex. As in the wireline case, a substantial effort has been invested in investigating stability in that setting. Stability in  wireless networks without explicit interferences was first studied by Andrews and Zhang \cite{AndrewsZ-JACM05,AndrewsZ07} and Cholvi and Kowalski~\cite{CholviK10}.   Chlebus et al.~\cite{ChlebusKR-TALG12} and Anantharamu et al.~\cite{ANANTHARAMU201927} studied adversarial broadcasting with interferences in the case of using single-hop radio networks. In multi-hop networks with interferences, Chlebus et al.~\cite{ChlebusCK-FOMC14} considered interactions among components of routing in wireless networks, which included transmission policies, scheduling policies and control mechanisms to coordinate transmissions with scheduling. In~\cite{7950907}, the authors demonstrated that there is no routing algorithm guaranteeing stability for an injection rate greater than $1/L$, where the parameter $L$ is the largest number of nodes which a packet needs to traverse while routed to its destination. They also provided a routing algorithm that guarantees stability for injection rates smaller than $1/L$.
Their approach, however, 
is not accurate for studying stability of longer-distance packets;
therefore, in this work we study how the stability of routing depends of the \emph{conflict graph} (which we will formally define later) of the underlying wireless networks, which is independent of the lengths of the packets' routes.

\paragraph*{Our results}

In this paper, we study the stability of dynamic routing in multi-hop radio networks with a specific methodology of adversarial traffic that reflects interferences. We focus on packet-oblivious routing protocols; that is, algorithms that do not take into account any historical information about packets or carried out by packets. Such protocols are well motivated in practice, as real forwarding protocols and corresponding data-link layer architectures are typically packet-oblivious.

\begin{table*}[ht]
\centering
\begin{tabular}{|c|c|c|c|}\hline
Routing & Scheduling  & Required & Maximum   \\
 algorithm &  policies (ALG) & knowledge &  injection rate  \\ \hline \hline

\textsc{USS-plus-Alg} & LIS, SIS & Bounds on the  &  $O(1/(e \cdot \Delta^H))$ \\ 
 & NFS, FTG &  number of links  &  \\ 
 &  &   and on the network's degree &  \\ \hline

\textsc{Coloring-plus-Alg} &  LIS, SIS & Full topology & $O(1/ \Delta^H)$   \\
 & NFS, FTG &  &   \\ \hline 
\end{tabular}
\caption{Summary of the paper's results.}
\label{tb:results}
\end{table*}


First, we give a new family of combinatorial structures, which we call \emph{universally strong selectors}, that are used to provide a set of transmission schedules. Making use of these structures, combined with some known queuing policies such as Longest In System (LIS), Shortest In System (SIS), Nearest From Source (NFS) and Furthest To Go (FTG), we  propose a {\em local-knowledge} packet-oblivious routing algorithm (i.e., which is working without using any global topological information) that guarantees stability for certain injection rates.  We show that this protocol is asymptotically optimal regarding the injection rate for which stability is guaranteed, mainly, for $\Theta(1/\Delta^H)$, where $\Delta^H$ is the maximum vertex degree of the conflict graph of the wireless network.

Later, we introduce a {\em packet-oblivious} routing algorithm that, by using the same queueing policies, guarantees stability for higher injection rates. This algorithm is optimal regarding the injection rate for which stability is guaranteed. However, it needs to use some global information of the system topology (so called {\em global-knowledge}). 

Table~\ref{tb:results} summarizes the main results of this paper.

The rest of the paper is structured as follows. Section~2 contains the technical preliminaries. In Section~3, we introduce and study universally strong selectors, which are the core components of the deterministic local-knowledge routing algorithm that is developed in Section~4, where we show that it is asymptotically optimal. In Section~5, we present a global-knowledge routing algorithm that guarantees universal stability for higher traffic,  we also show that it is  optimal. 
In  Section~6, we extend the results obtained for the Longest-In-System scheduling policy in Section 4 to other policies,
mainly, SIS, NFS and FTG. This extension is based on different technical tools, mainly, on reduction to the wired model with failures studied in~\cite{AlvarezBDSF05}, in which SIS, NFS and FTG are stable. 
We conclude with future directions in Section~7.  Some technical details and proofs are deferred to the Appendix.

\section{Model and Problem Definition}

\paragraph*{Wireless radio network}

We consider a {\em wireless radio network} represented by a directed symmetric network graph $G=(V_G, E_G)$. 
It consists of {\em nodes} in $V_G$ representing devices, and directed edges,
called {\em links}, representing the fact that a transmission
from the starting node of the link could be directly delivered to the ending node. 
The graph is symmetric in the sense that if some $(i,j) \in E_G$  then $(j,i) \in E_G$ too. 
Each node has a unique ID number and it  knows some upper bounds on the number $m$ of edges in the network and the network in-degree 
(i.e., the largest number of links incoming to a network node).\footnote{In which case the performance will depend on these known estimates, instead of the actual values.}

Nodes communicate via the underlying wireless network $G$.
Communication is in synchronous rounds.
In each round a node could be either transmitting or listening. 
Node $i$ receives a message from a node $j\ne i$ in a round if $j$ is the only transmitting in-neighbor of $i$ in this round and node $i$ does not transmit in this round;
we say that the message was successfully sent/transmitted from $j$ to $i$.

\paragraph*{Conflict graphs}

We define the {\em conflict graph} $H(G)=(V_{H(G)}, E_{H(G)})$ of a network $G$ as follows: (1) its vertices are 
links of the network (i.e., $V_{H(G)} = E_G$) and, (2) a directed edge $(u,v) \in E_{H(G)}$ if and only if a message across link $v \in E_G$ cannot be successfully
transmitted while link $u \in E_G$ transmits.
Note that, accordingly with the radio model, a conflict occurs if and only if the transmitter in $u$ is also a receiver in $v$ {\em or} the transmitter in $u$ is a neighbor of the receiver in $v$ (see Figure~\ref{fig:networks} for an illustrative example).
If network $G$ is clear from the context, we skip the parameter $G$ in $H(G)$ (i.e., we will use  $H$). 
Note that, the links in our definition are directed in order 
to distinguish which transmission is blocked by which.

\begin{figure}[t]
\centering
\includegraphics[scale=1]{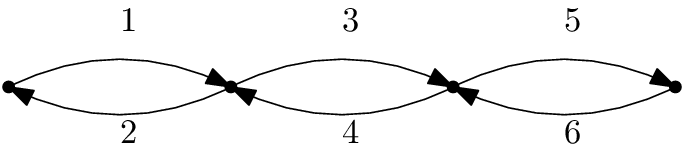} 

\vspace*{0.75cm}

\includegraphics[scale=1]{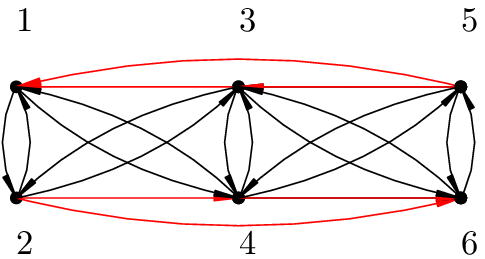}
\caption{Radio network $G$ with $4$ nodes and links labeled $1-6$ (up). Conflict graph $H(G)$ obtained from network $G$ (down). Observe that each link $i$ in network $G$ corresponds to one node $i$ in $H(G)$.}
\label{fig:networks}
\end{figure}

\paragraph*{Routing protocols and transmission schedules}

We consider {\em packet-oblivious} routing protocols,
that is, protocols which only use their hardwired memory
and basic parameters of the stored packets assigned to them at injection time (such as
source, destination, injection time, route) in order to decide
which packet to send and when.

We distinguish between {\em global-knowledge} protocols
which can use topological information given as input, and
{\em local-knowledge} protocols that are given only
basic system parameters such as the number of links or the network's in-degree.

All our protocols will be based on pre-defined transmission
schedules, which will be circularly repeated --- the properties
of these schedules will guarantee stability for certain
injection rates. These schedules will be different for different types of protocols, due to the available information based on which these schedules could be created.

\paragraph*{Adversaries}
\label{sec:adv}

We model dynamic injection of packets by way of an adversarial model, in the spirit of similar approaches used in~\cite{BorodinKRSW01,AndrewsAFLLK01,LotkerPR04,ChlebusKR-TALG12,CholviK10,ChlebusCK-FOMC14,7950907}. An adversary represents the users that generate packets to be routed in a given radio network. The constraints imposed on packet generation by the adversary allow considering worst-case performance of deterministic routing algorithms handling dynamic traffic.

Over time, an adversary injects packets to some nodes. The adversary decides on a path a packet has to traverse upon its injection. Our task is to develop a packet-oblivious routing protocol such that the network remains \emph{stable}; that is, the number of packets simultaneously queued is bounded by a constant in all rounds.
Since an unbounded adversary can exceed the capability of a network to transmit messages, we limit its power in the following way:
for any time window of any length $T$, the adversary can inject packets (with their paths) in such a way that each link is traversed by at most $\rho \cdot T + b$ packets, for some $0\leq \rho \leq 1$ and $b \in \mathbb{N}^+$. We call such an adversary a \emph{$(\rho, b)$-adversary}.

\section{Selectors as transmission schedulers}
\label{sec:selectors}

In this section, we introduce a family of combinatorial structures, widely called \emph{selectors}{~\cite{clementi2003,ChlebusKPR-ICALP11}}, that are the core of the deterministic routing algorithm presented in Section~\ref{sec:local_routing}. In short, we will use specific type of selectors to provide a set of transmission schedules that assure stability when combined with suitable queuing policies.

There are many different types of selectors, with the more general one being described below:

\begin{definition}
\label{def:selector}
Given integers $k, m$ and $n$, with $1 \leq m \leq k \leq n$, we say that a boolean matrix $M$ with $t$ rows and $n$ columns is a \emph{$(n, k, m)$-selector} if any submatrix of $M$ obtained by choosing $k$ out of $n$ arbitrary columns of $M$ contains at least $m$ distinct rows of the identity matrix $I_k$.  The integer $t$ is referred as the size of the $(n, k, m)$-selector.
\end{definition}

In order to use selectors as transmission schedules, the parameter $n$ is intended to refer to the number of nodes in the network, $k$ refers to the maximum number of nodes that can compete to transmit (i.e., $k=\Delta+1$, where $\Delta$ is the maximum degree of the network), and $m$ refers to the number of nodes that are guaranteed to successfully transmit during the $t$-round schedule. Therefore, each column of the matrix $M$ is used to define the whole transmission schedule of each node. Rows are used to decide which nodes should transmit at each time slot: In the $i$-th time slot, node $v$ will transmit iff $M_{i,v} = 1$ (and $v$ has a packet queued);  the schedule is repeated after each $t$ time slots.

Taking into account the above-mentioned approach, selectors may be used to guarantee that during the schedule, every node will successfully receive some messages. 

A $(n,k,1)$-selector guarantees that, for each node, one of its neighbors will successfully transmit during at least $1$ round per schedule cycle (that is, that node will successfully receive at least one message). However, whereas the above use of selectors is helpful in broadcasting (since there is progress every time any node receives a message from a neighbor), it happens that many neighbors may have something to send, but only one of them has something for that node. Therefore, the above presented selector guarantees that each node will receive at least one message, but not necessarily will receive the one addressed to it.

A $(n,k,k)$-selector (which is known as \emph{strong selector}~\cite{clementi2003}) guarantees that every node that has exactly $k$ neighbors will receive a message from each one of them. However, it has been shown that its size $t=\Omega(\min\{n, (k^2 / \log k) \log n\})$. This means that $k$ packets will be received, but during a long amount of time.

In order to solve the above mentioned problems with known selectors, now we introduce a new type of selectors, which we call \emph{universally strong}. Namely, a \emph{$(n,k,\epsilon)$-universally-strong} selector of length $t$ guarantees that every node will receive $\epsilon \cdot t /k$ successful messages from every neighbor during $t$ rounds. More formally:


\begin{definition}
A \emph{$(n,k,\epsilon)$-universally-strong selector} $\mathcal{S}$ is a family of $t$ sets $T_1,\dots,T_t\subseteq[n]$ such that for every set $A\subseteq[n]$ of at most $k$ elements and for every element $a \in A$ there exist at least $\epsilon \cdot t/k$ sets $T_i \in \mathcal{S}$ such that $T_i \cap A = \{a\}$.
\end{definition}

\subsection{Universally strong selectors of polynomial size}

Clearly, universally strong selectors make sense provided they exist and their size is moderate. In the next theorem, we prove that, for any $\epsilon \leq 1/e$, there exists a $(n,k,\epsilon)$-universally-strong selector of polynomial size.

\begin{theorem}
\label{theorem:selector_existence}
For any $\epsilon \leq 1/e$, there exists a $(n,k,\epsilon)$-universally-strong selector of size $O(k^2 \ln n)$.
\end{theorem}
\begin{proof}
The proof relies on the probabilistic method.

Consider a random matrix $M$ with $t$ rows and $n$ columns, where $M_{i,j} = 1$ with probability $p$ and $M_{i,j} = 0$ otherwise. Given a row $i$ and columns $j_1,\dots, j_k$, the probability that $M_{i, j_1}=1$ and $M_{i,j_2} = \dots = M_{i,j_k} = 0$ (i.e., that node $j_1$'s transmission is not interrupted by  nodes $j_2,\dots,j_k$ in round $i$) is $P = p(1-p)^{k-1}$ and is maximized with $p=1/k$. In further considerations, we use matrix $M$ generated with $p=1/k$.

Given columns $C= \{j_1,\dots, j_k\}$, let $X(C)$ be the number of ``good'' rows $i$ such that $M_{i, j_1}=1$ and $M_{i,j_2} = \dots = M_{i,j_k} = 0$.

We will use the following Chernoff bound: 
$$ Pr[X(C) \leq (1-\delta)E[X(C)]] \leq exp(-E[X(C)]\delta^2/2)$$
 for  $0 \leq \delta \leq 1. $ 

Using $E[X(C)]=Pt$ and $\delta = (kP-\epsilon)/(kP)$, we obtain:
$$ Pr[X(C) \leq \epsilon t / k] \leq exp(-Pt\delta^2/2). $$

Consider a ``bad'' event $\mathcal{E}$ such that for at least one set of columns of size at most $k$, there are few good rows. More specifically, $X(C) \leq \epsilon t/k$ for at least one set of columns $C$, where $|C| = k$. The probability $R$ of event $\mathcal{E}$ happening fulfills the following inequality:
$$ R \leq  k\binom{n}{k} exp(-Pt\delta^2/2). $$

Therefore $R < 1$ if 
$$ exp(-Pt\delta^2/2) < 1/ \left[k\binom{n}{k} \right] $$
$$ -Pt\delta^2/2 < -\ln \left(k\binom{n}{k}\right) $$
$$ Pt\left(\dfrac{kP-\epsilon}{kP}\right)^2/2 > \ln \left(k\binom{n}{k}\right) $$

Let $c=kP$. Using $\binom{n}{k} \leq \left(\dfrac{ne}{k}\right)^k$, provided $c \neq \epsilon$, we obtain the following:
$$ t(c-\epsilon)^2/(2ck) > \ln k + \ln \left(\dfrac{ne}{k}\right)^k $$
$$ t > \left[2ck \ln k + 2ck^2 \ln \left(\dfrac{ne}{k}\right)\right] / (c-\epsilon)^2$$

Therefore, as long as $0 \leq \delta = \frac{c-\epsilon}{c} \leq 1$ (so that we can use the Chernoff bound) and $\epsilon \neq c$, the probability of generating a random matrix $M$ such that event $\mathcal{E}$ occurs is less than $1$. Thus, there exists a matrix $M$ such that, for every set of $k$ columns $j_1,\dots, j_k$, there are at least $\epsilon t/k$ rows such that $M_{i, j_1}=1$ and $M_{i,j_2} = \dots = M_{i,j_k} = 0$. Trivially, such matrix $M$ guarantees the above property for any set of at most $k$ columns. Hence, $M$ represents a 
$(n,k,\epsilon)$-universally-strong selector, provided that $\epsilon < c = kP$. Next, we calculate which values of $\epsilon$ fulfill that inequality.

Consider a sequence $a_i = (1+1/i)^i$. $a_i$ is known as a lower bound on the Euler's number $e$ (i.e., $\forall i \; a_i < e$). Note that $c = kP = (1-1/k)^{k-1} = 1/a_{k-1} > 1/e$ for all $k \geq 2$. This implies that any $\epsilon \leq 1/e$ fulfills the requirement of $\delta > 0$ and results in the existence of a $(n,k,\epsilon)$-universally-strong selector.
\end{proof}

\subsection{Obtaining  universally strong selectors of polynomial size in polynomial time}
\label{subsec:poly}

In the proof of Theorem~\ref{theorem:selector_existence}, we have introduced a family of universally strong selectors of polynomial size. However, obtaining them by derandomizing would be very inefficient  (all the approaches we know are, at least, exponential in $n$). Here, we present an algorithm, which we call \textsc{Poly-Universally-Strong}, that computes universally strong selectors of polynomial size in polynomial time (they only have slightly lower values of $\epsilon$ comparing to the existential result in Theorem~\ref{theorem:selector_existence}).


The algorithm, whose code is shown in Figure~\ref{fig:selector}, has to be executed by each node in the network taking the same polynomials, so that all nodes will obtain exactly the same matrix that defines the transmission schedule.


\begin{figure}[t]
\begin{mdframed}
\begin{center}
\begin{enumerate}
\item
Let $d= \left \lceil \log_k n \right \rceil$ and $q=c \cdot k \cdot d$ for some constant $c>0$ such that $q^{d+1} \ge n$.

\item
Consider all polynomials $P_i$ of degree $d$ over field $[q]$. Notice that there are $q^{d+1}$ of such polynomials. 

\item
Create a matrix $M'$ of size $q \times q^{d+1}$.
Each column will represent values $P_i(x)$ of each polynomial $P_i$ for arguments $x = 0, 1, \dots, q-1$
(corresponding to rows of $M'$). 
Next, matrix $M''$ is created from $M'$ as follows: each value $y=P_i(x)$ is represented 
and padded in $q$ consecutive rows of $0$s and $1$s, where $1$ is on $y$-th position, while on all other positions there are $0$s. Notice that each column of $M''$ has $q^2$ rows ($q$ rows for each argument),
thus $M''$ has size $q^2\times q^{d+1}$. 

\item
Remove $q^{d+1}-n$ arbitrary columns from matrix $M''$, creating matrix $M$ with exactly $n$ columns remaining. 
\label{step:remove}

\item
Each row of matrix $M$ will correspond to one set $T_i$ of a universally strong selector $\{T_i\}_{i=1}^{q^2}$
over the set $\{1,\ldots,n\}$ of elements.

\end{enumerate}

\caption{The \textsc{Poly-Universally-Strong} algorithm, given parameters $n$ and $k$.}
\label{fig:selector}
\end{center}
\end{mdframed}
\end{figure}


The next theorem shows that, indeed, it constructs a $(n,k,\epsilon)$-universally-strong selector of polynomial size with $\epsilon= 1/(4 \log_k n)$. 

\begin{theorem}
\label{th:poly}
\textsc{Poly-Universally-Strong}\ constructs (by using $c=2$) a $(n,k,\epsilon)$-universally-strong selector of size $4 \cdot k^2 \cdot \left\lceil\log_k n \right\rceil^2$ with $\epsilon = 1/(4 \log_k n)$. 
\end{theorem}

\begin{proof}
First, note that two polynomials $P_i$ and $P_j$ of degree $d$ with $i \neq j$, can have equal values for at most $d$ different arguments. This is because they have equal values for arguments $x$ for which $P_i(x) - P_j(x)=0$. However, $P_i - P_j$ is a polynomial of degree at most $d$, so it can have at most $d$ zeroes. So, $P_i(x) = P_j(x)$ for at most $d$ different arguments~$x$.

Take any polynomial $P_i$ and any $k$ polynomials $P_j$ still represented in $M$ (so excluding columns/polynomials removed from consideration in step~{\ref{step:remove}}). There are at most $k \cdot d$ different arguments where one of the $k$ polynomials can be equal to $P_i$. So, for $q-k \cdot d$ different arguments, the values of the polynomial $P_i$ are unique. Therefore, if we look at rows with $1$ in column $i$ of matrix $M$ (there are $q$ of those rows, one for each argument), at least $q-k \cdot d$ of them have $0$s in chosen $k$ columns. Since there are $q^2$ rows, so a fraction $(q-k \cdot d)/q^2$ of rows have the desired property (i.e., there is value $1$ in column $i$ and value $0$ in the chosen $k$ columns):

$$\dfrac{q-k \cdot d}{q^2} = \dfrac{(c-1) \cdot k \cdot d}{(c \cdot k \cdot d)^2} = \dfrac{c-1}{c^2 \cdot k \cdot d} \triangleq f(c) \ . $$

Let us find the value of $c$ that maximizes the function $f$. To do it, we compute its differential 

$$f'(c) = (\dfrac{c-1}{c^2 \cdot k \cdot d})' = \dfrac{1\cdot(c^2 \cdot k \cdot d)-(c-1)\cdot k \cdot d \cdot 2c}{c^4 \cdot k^2 \cdot d^2} =$$ 
$$= \dfrac{-c^2 \cdot k \cdot d + 2c \cdot k \cdot d}{c^4 \cdot k^2 \cdot d^2} = \dfrac{-c + 2}{c^3 \cdot k \cdot d}.$$ 

Thus, $f'(c)=0$ for $c=0$ or $c=2$. The value $c=2$ maximizes $f$, giving 
$f(c)\le f(2) = 1/(4k \cdot d) = 1/(4k \cdot \log_k n)$.

Therefore, we can construct a $(n,k,\epsilon)$-universally-strong selector with $\epsilon = f(2) \cdot k = 1/(4d) = 1/(4 \log_k n)$ of length $4k^2 \cdot \left\lceil\log_k n \right\rceil^2$ (which means that an $f(2) = 1/(4k \cdot \log_k n)$ fraction of the selector's sets have the desired property).
\end{proof}


\section{A local-knowledge routing algorithm}
\label{sec:local_routing}

In this section, we introduce a local-knowledge packet-oblivious routing algorithm that makes use of the family of  universally strong selectors introduced in Section~\ref{sec:selectors} as \emph{transmission schedules} (i.e., the time instants when packets stored at each one node must be transmitted to a receiving node). As it has been mentioned previously, local-knowledge routing algorithms work without using any topological information, except for maybe some network's features that do not require full knowledge of its topology. In our particular case, that will consists of some upper bounds on the number of links and on the network's degree.

The code of the proposed algorithm, which we call \textsc{USS-plus-Alg}, is shown in Figure~\ref{fig:routing}. Given a graph $G$ with a number of links bounded by $m$, and an in-degree of its conflict graph $H$  (which we denote as $\Delta_{in}^H$) bounded by $\Delta \geq 1$, it uses a $(m,\Delta+1,\epsilon)$-universally-strong selector as a schedule: assuming the selector is represented by matrix $M$ with $t$ rows, each link $z \in E_G$ will transmit at time $i$ iff $M_{i \mod t, z}=1$. Notice that here each link is assumed to have an independent queue, and therefore they will act as a sort of ``nodes'' (in terms of selectors, such as it has been stated in the previous section). This means that each individual link will have its own schedule.

\begin{figure}[t]
\begin{mdframed}
\begin{center}
\begin{enumerate}
\item
Choose $m$ and $\Delta$  such that $\mid E(G) \mid  \; \leq m$ and $\Delta_{in}^H \leq \Delta$.

\item
Obtain a $(m,\Delta + 1,\epsilon)$ universally strong selector (for some value of $\epsilon$) of some length $t$ and use it as the transmission schedule.

\item
When there are several packets awaiting in a single queue, choose the packet to be transmitted according to \textsc{Alg}, breaking ties in any arbitrary fashion.

\end{enumerate}
\caption{The \textsc{USS-plus-Alg} algorithm for a network $G$.}
\label{fig:routing}
\end{center}
\end{mdframed}
\end{figure}

\subsection{The \textsc{USS-plus-LIS} algorithm}

Next, we show that  the \textsc{USS-plus-LIS} algorithm (i.e., the  \textsc{USS-plus-Alg} algorithm where \textsc{Alg} is the Longest-In-System scheduling policy), guarantees stability, provided a given packets' injection admissibility condition is fulfilled. But first, we define what is a $(\rho',T)$-frequent schedule.

\begin{definition}
\emph{A $(\rho,T)$-frequent schedule} for graph $G$ is an algorithm that decides which links of graph $G$ transmit at every round in such a way that each link is guaranteed to successfuly (without radio network collisions) transmit at least $\rho \cdot T$ times in any window of length $T$ (provided at least $\rho \cdot T$ packets await for transmission at the link at the start of the window).
\end{definition}

At this point, we note that the transmission schedules provided by our universally strong selectors can be seen as \emph{$(\rho,T)$-frequent schedules}. \\

We now proceed with the main result in this section.

\begin{theorem}
\label{the:local}
Given a network $G$, the \textsc{USS-plus-LIS} algorithm is stable against any $(\rho,b)$-adversary, for $\rho < \frac{\epsilon}{\Delta+1}$.
\end{theorem}
\begin{proof}
Let us take any arbitrary link $z\in E_G$ and consider the set of all other links that conflict with link $z$, of which there are at most $\Delta$. This means that there exist at least $\epsilon \cdot  t/(\Delta+1)$ rows $i$ in $M$ such that $M_{i \mod t ,z}=1$ and $M_{i \mod t,c_1} = \dots = M_{i \mod t,c_j} = 0$. Therefore, at time $i$, link $z$ will transmit a message, and no link that conflicts with the link $z$ will transmit. This guarantees that each link will successfully transmit, at least, $\epsilon \cdot t/(\Delta+1)$ messages during any schedule of length $t$ (i.e., we obtained a $(\epsilon/(\Delta+1),t)$-frequent schedule $\mathcal S$). Then, we can apply the result Lemma~\ref{lem:LIS}  in the Appendix~\ref{app:auxiliary} to deduce that such an algorithm is stable against any $(\rho,b)$-adversary, where $\rho < \frac{\epsilon}{\Delta+1}$. 
\end{proof}

By using the selectors provided by the  \textsc{Poly-Universally-Strong}\ algorithm in \textsc{USS-plus-LIS}, we have the following result:

\begin{corollary}
Given a network $G$, the \textsc{USS-plus-LIS} algorithm using a universally strong selector computed by the \textsc{Poly-Universally-Strong} algorithm is a stable algorithm against any $(\rho,b)$-adversary, for $\rho < \frac{1}{4(\Delta+1)\log_{\Delta+1}m }$.
\end{corollary}

If instead of the selectors provided by the \textsc{Poly-Universally-Strong}\ algorithm, we use a selector from Theorem~\ref{theorem:selector_existence}, we have that:

\begin{corollary}
\label{corollary:best_selector}
Given a network $G$, there exists a universally strong selector that, used in the \textsc{USS-plus-LIS} algorithm, provides a stable algorithm against any $(\rho,b)$-adversary, for $\rho < \frac{1}{e \cdot (\Delta+1)}$.
\end{corollary}

As it can be readily seen, the \textsc{USS-plus-LIS} algorithm for a network $G$ requires some knowledge of the value of the in-degree of its conflict graph $H$ (i.e., of $\Delta_{in}^H$). In order to obtain $H$ it is necessary to gather the whole topology of $G$.   However, as the next lemma shows, $\Delta_{in}^H$ can be bounded by the in-degree of the network $G$ (denoted $\Delta_G$).

\begin{lemma}
\label{lem:GvsH}
$\Delta_{in}^H \leq \Delta_G^2 + \Delta_G -1$, provided $\Delta_G >0$.
\end{lemma}
\begin{proof}
If $\Delta_{in}^H=0$, then the lemma is trivially true. Otherwise, consider a vertex $e$ in $H$ of maximum in-degree $deg(e) = \Delta_{in}^H$. Since $\Delta_{in}^H \neq 0$, there is at least one edge $(e',e) \in H$ such that, in $G$,  $e$ cannot successfully transmit at the same time instant when $e'$ transmits. Let us denote $e=(u,v)$ and $e'=(u',v')$, and let us consider the different scenarios where $e$ and $e'$ may conflict.

Now, we make a case analysis regarding the possible conflicts in $G$ (note that its in-degree is equal to its  out-degree, since $G$ is symmetric):
\begin{enumerate}
\item
$u' = u$ and $v' \neq v$ (a node $u=u'$ cannot transmit messages to 2 different receivers): there are at most $\Delta_G-1$ such links $e'$, given fixed link $e$.

\item
$u' = v$ (if $u'$ transmits, it cannot listen at the same time): there are at most $\Delta_G$ such links $e'$, given fixed link $e$.

\item
$u' \neq u$ is a neighbor of $v$ (i.e., $v$ can hear both from $u$ and $u'$): there are at most $\Delta_G-1$ neighbors of $v$ different than node $u$, and each of them has, at most, $\Delta_G$ different links. This gives $\Delta_G^2 - \Delta_G$ such links $e'$, given fixed link $e$.
\end{enumerate}

Therefore, in overall there are at most $(\Delta_G-1) + \Delta_G  + (\Delta_G^2 - \Delta_G) = \Delta_G^2 + \Delta_G -1$ such links. 
\end{proof}


The previous lemma shows that \textsc{USS-plus-LIS} can be seen as a local-knowledge algorithm, in the sense that it only requires some knowledge about two basic system parameters: the number of links and the network's in-degree.  In Section~\ref{sec:coloring}
we will also look at a solution that 
requires some  global-knowledge of $G$.

\subsection{Optimality of the {\sc USS-plus-LIS} algorithm}


In the next theorem, we show an impossibility result regarding routing algorithms, either based on selectors or not, that only make use of upper bounds on the number of links and on the network's degree.\footnote{To be strict, it is also necessary that each node $v$ decides whether or not to an outgoing edge $e=(u,v)$ should transmit at time $t$ based on $t$ and on the ID of node $u$.}

\begin{theorem}
\label{the:impossibility}
No routing algorithm that only makes use of upper bounds on the number of links and on the network's degree guarantees stability for all networks of degree at most $\Delta$, provided the injection rate $\rho = \omega(1 / \Delta^2)$.
\end{theorem}
\begin{proof}
See Appendix~\ref{app:impossibility}.
\end{proof}

If we apply Theorem~\ref{the:impossibility} to \textsc{USS-plus-LIS}, then our goal  is to find how close to $\rho = O(1/\Delta_G^2)$ is its maximum injection rate for which it guarantees stability.

If we consider Theorem~\ref{the:local} with $\Delta=\Delta_{in}^H$, we have that \textsc{USS-plus-LIS} can be stable for $\rho = O(1/\Delta_{in}^H)$. Furthermore, by Lemma~\ref{lem:GvsH} we know that $\Delta_{in}^H$ can be as large as $\Theta(\Delta_G^2)$. Then, we have that \textsc{USS-plus-LIS}  guarantees stability for $\rho = O(1/\Delta_G^2)$  for all networks $G$, which matches the result in Theorem~\ref{the:impossibility}. This proves that the \textsc{USS-plus-LIS} algorithm is asymptotically optimal regarding the injection rate for which stability is guaranteed.

\section{A global-knowledge routing algorithm}
\label{sec:coloring}

In this section, we introduce a global-knowledge packet-oblivious routing algorithm, which we call \textsc{Coloring-plus-Alg}, that is based on using graph coloring as \emph{transmission schedules}. Such an algorithm  does not take into account any historical information. However, it has to be seeded by some information about the network topology (i.e., it is a global-knowledge routing algorithm).

Next, we show that  the \textsc{Coloring-plus-LIS} algorithm (i.e., the  \textsc{Coloring-plus-Alg} algorithm where \textsc{Alg} is the Longest-In-System scheduling policy), guarantees stability, provided a given packets' injection admissibility condition is fulfilled.

Now, we proceed withe the main results in this section. But before we introduce the \textsc{Coloring-plus-Alg} routing algorithm, we state the following fact regarding the relationship between vertex coloring in a conflict graph, and its use as a transmission schedule.

\begin{fact}[\cite{7950907}]
	Vertex coloring of the conflict graph $H(G)$ using $x$ colors is equivalent to a schedule of length $x$ that successfully transmits a packet via each directed link of network $G$. 
\end{fact}

Note that every set of vertices of same color can be extended to a maximal independent set. The resulting family of independent sets is still a feasible schedule that guarantees no conflicts and is no worse than just coloring. In fact, it may allow some links to transmit more than once during the schedule, without increasing the length of the schedule.

Following, we show that coloring of a collision graph can be used to obtain a transmission schedule, where each link is guaranteed to regularly transmit.

\begin{lemma}
\label{lem:schedule}
A $k$-coloring of collision graph $H$ provides a $(1/k, k)$-frequent schedule.
\end{lemma}

\begin{proof}
Let us split the vertices $V_H$ of the graph $H$ into sets $V_H^i$ for $i = 0,1,\dots,k-1$, where every vertex in $V_H^i$ is assigned the $i$-th color in the vertex coloring of graph $H$. Each link in the graph $G$ is represented by one vertex in $V_H$, and therefore each link is assigned a unique color. According to the definition of the conflict graph $H$, if there is no edge $(u,v) \in E_H$, then links $u \in E_G$ and $v \in E_G$ can deliver their packets simultaneously, without a collision. Therefore, if at a given round $t$ only links of $(t \mod i)$-th color transmit, then no collision occurs. Since each link has a color $i \in \{0,1,\dots, k-1\}$ assigned to it, then each link will successfully transmit a packet once each $k$ consecutive rounds (as far as there is one packet waiting in its queue). 
\end{proof}

Since $\chi(H)$-coloring is an optimal coloring of graph $H$, we have the following result.

\begin{corollary}
\label{cor:schedule}
An optimal coloring of collision graph $H$ provides a $(1/\chi(H), \chi(H))$-frequent schedule.
\end{corollary}

Once we have made it clear  that coloring of a collision graph can be used to obtain a transmission schedule, the code of the \textsc{Coloring-plus-Alg} algorithm is shown in Figure~\ref{fig:global}.

\begin{figure}[t]
\begin{mdframed}
\begin{center}
\begin{enumerate}
\item
Use optimal coloring of graph $H$ as the transmission schedule, and repeat it indefi\-nitely.
\item
When there are several packets awaiting in a single queue, choose the packet to be transmitted according to \textsc{Alg},  breaking ties in any arbitrary fashion.
\end{enumerate}
\caption{The \textsc{Coloring-plus-Alg} algorithm for graph $G$.}
\label{fig:global}
\end{center}
\end{mdframed}
\end{figure}

\subsection{The \textsc{Coloring-plus-LIS} algorithm}

Now, we show that \textsc{Coloring-plus-LIS} (i.e., the  \textsc{Coloring-plus-Alg} algorithm where \textsc{Alg} is the Longest-In-System scheduling policy), guarantees stability, provided a given packets' injection admissibility condition is fulfilled.

\begin{theorem}
\label{the:coloring}
The \textsc{Coloring-plus-LIS} algorithm is stable provided $\rho < 1/\chi(H)$, where $\chi(H)$ is the chromatic number of the conflict graph $H$ of the network $G$.
\end{theorem}

\begin{proof}
We start the proof with referring to Corollary~\ref{cor:schedule}, which shows that coloring of a collision graph can be used to obtain a $(1/\chi(H),\chi(H))$-frequent schedule $\mathcal{C}$.

Let us take any $\rho = 1/\chi(H) - \epsilon$, for some $\epsilon>0$. We can use Lemma~\ref{lem:LIS} in the Appendix~\ref{app:auxiliary} with $\mathcal{S}=\mathcal{C}$ (so, $\rho' = 1/\chi(H)$) to show that \textsc{Coloring-plus-LIS} is stable against any $(\rho,b)$-adversary in the radio network model. 
\end{proof}

Observe that, contrary to the {\sc USS-plus-LIS} protocol, the \textsc{Coloring-plus-LIS} algorithm requires global-knowledge of the structure of the graph: first, to construct $H$, and then to obtain its optimal coloring.

\subsection{Optimality of the  {\sc Coloring-plus-LIS} algorithm}

Now, we show that the  {\sc Coloring-plus-LIS} algorithm is optimal regarding the injection rate, in the sense that no algorithm can guarantee stability for a higher injection rate that that provide by it.

\begin{theorem}
\label{theorem:optimality_coloring}
No algorithm can be stable for all networks against a $(\rho,b)$-adversary for $\rho>1/\chi(H)$.
\end{theorem}
\begin{proof}
See Appendix~\ref{app:optimality_coloring}.
\end{proof}

\subsection{Global-knowledge vs local-knowledge  routing  protocols}

Regarding the {\sc Coloring-plus-LIS} protocol, by the Brooks' theorem~\cite{Brooks1941}, we have that $\chi(H) \leq \Delta^H +1$. Let $indeg^H(e)$ (and $outdeg^H(e)$) denote the indegree (outdegree) of node $e$ in graph $H$. Recall that each edge in the network graph was replaced by two oppositely directed links. This means that, if a link $e$ blocks $outdeg^H(e)$ other links, then the opposite link $e'$ is blocked by  $ indeg^H(e') = outdeg^H(e)$ links.
Therefore, $\Delta^H = \Theta(\Delta^H_{in})$.
 Then, Theorem~\ref{the:coloring} guarantees stability for $\rho=O(1/\Delta^H_{in})$. 

On the other hand, from the result in Corollary~\ref{corollary:best_selector}, we have that \textsc{USS-plus-LIS} can only guarantee stability for $\rho=O(1/(e \cdot \Delta_{in}^H))$. This implies that, by using the {\sc Coloring-plus-LIS} protocol, it is possible to guarantee stability for a wider range of injection rates than by using the {\sc USS-plus-LIS} protocol: namely, the injection rate for which stability is guaranteed is $e$ times higher.


\section{Extension of the results to other scheduling policies}
\label{sec:extension}

In this section, we show that the results obtained in sections~\ref{sec:coloring} and~\ref{sec:local_routing} for routing combined with LIS (Longest In System) can be extended to other scheduling policies; namely, NFS (Nearest-From-Source), SIS (Shortest-In-System) and FTG (Farthest-To-Go). Indeed, for such a scheduling policies, Theorems~\ref{the:par_local} and~\ref{the:par_global} respectively parallelize the analogous results in Theorems~\ref{the:local} and~{\ref{the:coloring} obtained for LIS.

\begin{theorem}
\label{the:par_local}
Given a network $G$, the \textsc{USS-plus-Alg} algorithm, where $\textsc{Alg}  \in \{ \textrm{NFS,SIS,FTG} \}$, is stable against any $(\rho,b)$-adversary, for $\rho < \frac{\epsilon}{\Delta+1}$.
\end{theorem}
\begin{proof}
The proof is similar to that in Theorem~\ref{the:local}. The only difference is that, instead of Lemma~\ref{lem:LIS}, we can apply the results in Lemma~\ref{lem:reduction} for NFS, SIS and FTG  (see Appendix~\ref{app:reduction}) to deduce that such an algorithm is stable against any $(\rho,b)$-adversary, where $\rho < \frac{\epsilon}{\Delta+1}$. 
\end{proof}

\begin{theorem}
\label{the:par_global}
The \textsc{Coloring-plus-Alg} algorithm, where $\textsc{Alg}  \in \{ \textrm{NFS,SIS,FTG} \}$, is stable provided $\rho < 1/\chi(H)$, where $\chi(H)$ is the chromatic number of the conflict graph $H$ of the network $G$.
\end{theorem}
\begin{proof}
We will reduce the packet scheduling in radio network problem to the problem of packet scheduling in the wired failure model~\cite{AlvarezBDSF05}, in which these policies are known to be stable.

We start the proof with referring to Corollary~\ref{cor:schedule}, which shows that coloring of a collision graph can be used to obtain a $(1/\chi(H),\chi(H))$-frequent schedule $\mathcal{C}$.

Let us take any $\rho = 1/\chi(H) - \epsilon$, for some $\epsilon>0$. Now, we can use Lemma~\ref{lem:reduction} with $\mathcal{S}=\mathcal{C}$ (so, $\rho' = 1/\chi(H)$) and $\textsc{Alg}  \in \{\textrm{NFS,SIS,FTG}\}$ (with $\rho''=1-\epsilon$) to show that we can build an algorithm that is stable against any $(\rho,b)$-adversary in the radio network model (see Appendix~\ref{app:reduction}). Note that \textsc{Coloring-plus-Alg} is a special case of the algorithm built in the proof of Lemma~\ref{lem:reduction} with $\mathcal{S}=\mathcal{C}$. Therefore \textsc{Coloring-plus-Alg} with $\textsc{Alg}  \in \{\textrm{NFS,SIS,FTG}\}$ is stable against any $(\rho,b)$-adversary in the radio network model.  
\end{proof}

\section{Conclusions and future work}
\label{sec:conclusions}

In this work, we studied the fundamental problem of stability in multi-hop wireless networks.

We introduced  a new family of combinatorial structures, which we call \emph{universally strong selectors}, that are used to provide a set of transmission schedules. Making use of these structures, combined with some known queuing policies, we  propose a packet-oblivious routing algorithm which is working without using any global topological information, and guarantees stability for certain injection rates.  We show that this protocol is asymptotically optimal regarding the injection rate for which stability is guaranteed.

Furthermore, we also introduced a packet-oblivious routing algorithm that guarantees stability for higher traffic. We also show that this protocol is optimal regarding the injection rate for which stability is guaranteed.
However, it needs to use some global information of the system topology.

A natural direction would be to study other classes of protocols; for instance, when packets are injected without pre-defined routes. 
Universally strong selectors are interesting on its own right -- finding more applications for them is a promising open direction.
Finally, exploring the reductions between various settings of adversarial routing could lead to new discoveries, as demonstrated in the last part of this work.

\bibliographystyle{plain}
\bibliography{packet_oblivious_arxiv_new1}

\appendix

\section{Auxiliary Lemma}
\label{app:auxiliary}

Here, we show that LIS, combined with a transmission schedule that guarantees a number of successful transmissions in some time interval, makes the resulting routing protocol stable (these lemmas  are adapted versions of analogous results about universal stability of the \textsc{LIS} protocol in wired network~\cite{AndrewsAFLLK01}).  

\begin{lemma}
\label{lem:LIS}
If there exists a $(\rho',T)$-frequent schedule $\mathcal{S}$, then using LIS as the queueing policy guarantees stability of the resulting routing protocol against any $(\rho, b)$-adversary for $\rho < \rho'$.
\end{lemma}

Before we prove this lemma, we will introduce some additional notations and auxiliary lemmas.

Let $L$ be the length of the longest route in the system. Let us denote by class $i$ the set of packets injected during $i$-th window. A class $i$ is said to be \emph{active} during a window $w$ if and only if at some time during window $w$ there is some packet in the system of class $i' \leq i$.

Consider some packet $p$ injected during window $W_0$, whose path crosses links $e_1, e_2, \dots, e_L$, in this order. We use $W_i$ to denote the window, during which $p$ crossed link $e_i$. Let $c_w$ denote the number of active classes during window $w$. We define $c = \max_{ w \in [W_0,W_L) } c_w$. Then, we can bound the number of windows to deliver $p$.

\begin{lemma}
\label{lem:delivery}
$$W_L - W_0 \leq \frac{1-\left(1-\frac{\rho}{\rho'}\right)^L}{\rho \cdot T}\cdot (b-1) + c \cdot \left[1-\left(1-\frac{\rho}{\rho'}\right)^L\right] \ . $$
\end{lemma}
\begin{proof}

The packet $p$ reaches link $e_i$ for the first time in window $W_{i-1}$. Since $p$ is in the system, during window $W_{i-1}$ all classes $[W_0,W_{i-1}]$ are active. Therefore, according to the definition of $c$, there are at most $c - (W_{i-1} - W_0)$ active classes with packets older than packet $p$. Packets in those classes are the only packets that take priority over packet $p$ on link $e_i$. The oldest such packet was injected during window $w_{first} = W_0 - [c - (W_{i-1} - W_0)] = W_{i-1} - c$. Since its injection, at most $(W_0 - w_{first}) \cdot \rho \cdot T + b = [c - (W_{i-1} - W_0)] \cdot \rho \cdot T + b$ packets older than $p$ could be injected into the system. Therefore, there are at most $ [c - (W_{i-1} - W_0)] \cdot \rho \cdot T + b - 1 $ packets that will take priority over packet $p$ on link $e_i$. Since each link transmits at least $\rho' T$ times per window, the number of windows until $p$ transgresses link $e_i$ is at most
$$W_i - W_{i-1} \leq  \frac{\rho \cdot T \cdot (c + W_0 - W_{i-1}) + b -1}{\rho' \cdot T} \ .$$
Hence, 
$$W_i \leq \left(1-\frac{\rho}{\rho'}\right) W_{i-1} + \frac{\rho}{\rho'} \left(c + W_0\right) + \frac{b-1}{\rho' \cdot T} \ .$$

Therefore, solving the recurrence, we get:

$$W_L   \leq   W_0 + c\left[ 1-(1-\frac{\rho}{\rho'})^L \right] + \frac{1-(1-\frac{\rho}{\rho'})^L}{\rho \cdot T} (b-1),$$
which proves the lemma. 
\end{proof}

Now we have a bound on how long packet $p$ can be in the system, depending on value $c$. We will show that $c$ is bounded by a constant, depending only on network and adversary parameters, i.e., $L$, $\rho$ and $b$, and value $\rho'$ from Lemma~\ref{lem:LIS}.

\begin{lemma}
\label{lemma:classes}
There are never more than 
$$ (b-1) \cdot \frac{1-(1-\frac{\rho}{\rho'})^L}{(1-\frac{\rho}{\rho'})^L \cdot \rho \cdot T}$$
active classes in the system.
\end{lemma}
\begin{proof}
Let $c' =(b-1) \cdot \frac{1-(1-\frac{\rho}{\rho'})^L}{(1-\frac{\rho}{\rho'})^L \cdot \rho \cdot T} + \frac{1}{(1-\frac{\rho}{\rho'})^L}$. Assume, by contradiction, that a window $w$ is the first window during which there are at least $c'+1$ active classes. Hence, at the end of window $w-1$, there is a packet $q$ that was in the system for $c'$ windows, and no more than $c'$ classes were active until the end of window $w-1$.

According to Lemma~\ref{lem:delivery}, packet $q$ is delivered in at most 
$$c'\left[ 1-(1-\frac{\rho}{\rho'})^L \right] + \frac{1-(1-\frac{\rho}{\rho'})^L}{\rho \cdot T} (b-1) = $$ 
$$ = c'\left[ 1-(1-\frac{\rho}{\rho'})^L \right] + \left(c' - \frac{1}{(1-\frac{\rho}{\rho'})^L}\right) \cdot (1-\frac{\rho}{\rho'})^L =  $$
$$ = c' - 1 $$
windows, which gives a contradiction.
\end{proof}

Now that we have proven that any packet $p$ spends bounded time in the system, we can prove Lemma~\ref{lem:LIS}.

\begin{proof}[Proof of Lemma~\ref{lem:LIS}]
In Lemma~\ref{lemma:classes}, it has been shown that $c$ is bounded. By Lemma~\ref{lemma:classes}, this implies that $W_L - W_0$ is also bounded. This result guarantees that each packet spends a bounded time in the system. That means that such system is stable against any $(\rho,b)$-adversary, provided that $\rho' > \rho$, which completes the proof of the lemma. 
\end{proof}

\section{Proof of Theorem~\ref{the:impossibility}}
\label{app:impossibility}

\newtheorem*{the:impossibility}{Theorem \ref{the:impossibility}}

\begin{the:impossibility}
No routing algorithm that only makes use of upper bounds on the number of links and on the network's degree guarantees stability for all networks of degree at most $\Delta$, provided the injection rate $\rho = \omega(1 / \Delta^2)$.
\end{the:impossibility}
\begin{proof}
Assume, to the contrary, that there exists a routing algorithm $ALG$ such that, given any network of which it is aware of both its number of links and its degree, it guarantees that there are no more than $Q_{max}$ packets in the system at all times against all adversaries with injection rate $\rho = \omega(1 / \Delta^2)$.
Note that $Q_{max}$ could be a function on $\rho,n$, but a constant with respect to time.


\begin{figure}
\centering
\includegraphics[scale=0.5]{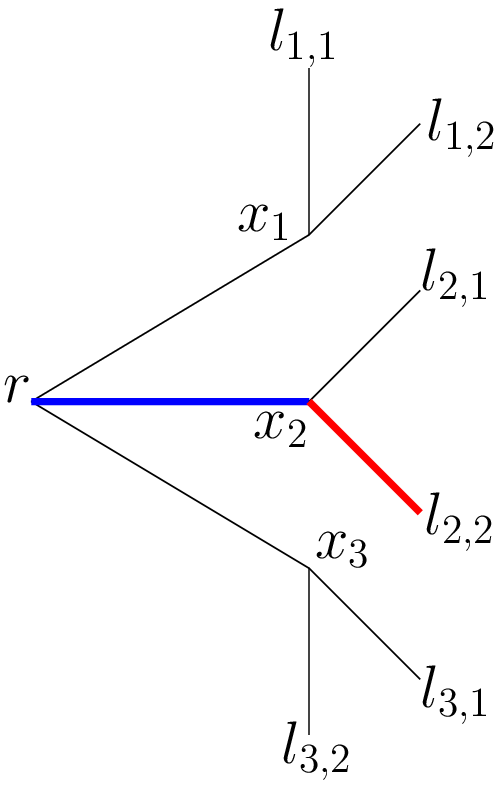}
\hspace{10mm}
\includegraphics[scale=0.5]{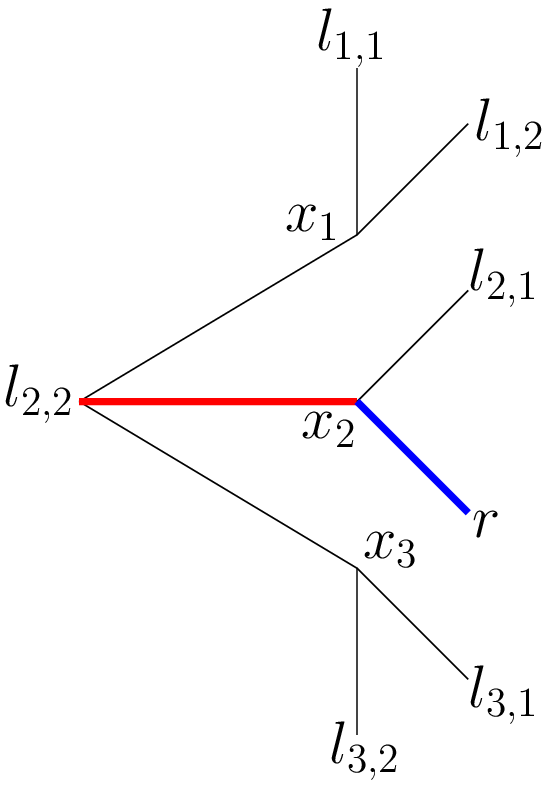}
\caption{Example of a tree $T$ (on the left) and tree $T_{2,2}$ (on the right) for $\Delta=3$. Nodes $r$ and $l_{2,2}$ swapped places, which means that edges $(x_2, r)$ and $(x_2,l_{2,2})$ (marked in blue and red, respectively) swapped their places as well.}
\label{fig:optimality}
\end{figure}

Consider a complete $\Delta$-regular tree $T$ of depth 2, rooted at $r$. Let us denote the nodes at distance $1$ from $r$ as $x_i$, for $i=1,\dots,\Delta$ and leaves adjacent to $x_i$ as $l_i^j$ for $j=1,\dots,\Delta-1$. Let us generate a family $\cF$ of trees $T_{i,j}$ as follows: swap the root $r$ of $T$ with leaf $l_{i,j}$ of $T$ (see Figure~\ref{fig:optimality}).
Note that edges $(x_i,r)$ and $(x_i,l_{i,j})$ swapped places, edges $(x_k,r)$ for $k\neq i$ were removed and in their place edges $(x_k,l_{i,j})$ appeared. Other edges, i.e., $(x_k, l_{k,a})$ for $a=1,\dots,\Delta-1$ and $(x_i, l_{i,b})$ for $b\neq j$, remain in the same place in both $T$ and $T_{i,j}$.

Note that edges $(x_k, l_{k,a})$ (for $k=1,\dots,\Delta$ and $a=1,\dots,\Delta-1$) exist in every tree in $\cF \cup \{T\}$. Let us denote the set of these edges as $E$.

Consider an adversary ${\cal A}$ that, starting from round $0$, injects $1$ packet into every edge outgoing from $x_i$ (for $i=1,\dots,\Delta$) every $1 / \rho$ rounds. Such adversary is a $(\rho,1)$-adversary in each tree in $\cF \cup \{T\}$.

Note that each packet injected into an edge incoming into the root of a tree $T' \in \cF \cup \{T\}$ cannot be simultaneously transmitted with any other packet injected by ${\cal A}$. In particular, it cannot be simultaneously transmitted with any other packet on edges in $E$.

Consider a time prefix of length $\tau$ rounds. Consider any edge $e \in E$. Edge $e$ is incident to the root in some tree $T' \in \cF \cup \{T\}$. $ALG$ must successfully transmit from $e$ in $T'$ in at least $\rho \tau - Q_{max}$ rounds during the considered prefix, since \textsc{ALG} is stable. This means that all other edges in $E$ must not transmit in those rounds. Since there are $\Delta (\Delta-1)$ possible choices of edge $e \in E$, each choice requiring all other edges in $E$ not to transmit in $\rho \tau - Q_{max}$ rounds, we get that each edge in $E$ must not transmit in $\Delta (\Delta-1) \cdot (\rho \tau - Q_{max})$ rounds and must transmit in $\rho \tau - Q_{max}$ rounds, for a total of  $\Delta^2 \cdot (\rho \tau - Q_{max})$ rounds in the prefix of length $\tau$. For $\rho = \omega(1/\Delta^2)$, we can choose $\tau$ such that $\Delta^2 \cdot (\rho \tau - Q_{max}) > \tau$, which gives us a contradiction. 
\end{proof}

\section{Proof of Theorem~\ref{theorem:optimality_coloring}}
\label{app:optimality_coloring}

\newtheorem*{theorem:optimality_coloring}{Theorem \ref{theorem:optimality_coloring}}

\begin{theorem:optimality_coloring}
No algorithm can be stable for all networks against a $(\rho,b)$-adversary for $\rho>1/\chi(H)$.
\end{theorem:optimality_coloring}
\begin{proof}
Let us consider a network graph $G$ on $n$ nodes that is a clique. For such network, the collision graph $H$ is also a clique, since each link is in conflict with each other link. Collision graph $H$ has $n^2-n$ vertices and requires $n^2-n$ colors to be colored, i.e., $\chi(H) = n^2-n$.

Consider a $(1/\chi(H) + \varepsilon, 2)$-adversary for some $\varepsilon>0$ that after every $\chi(H)$ rounds injects one packet into each link (starting in round $0$) and simulateously after each $1/\varepsilon$ rounds injects another packet into each link (starting in round $0$). Therefore, in any prefix of $T = k \cdot \chi(H)$ rounds for $k \in \mathbb{N}$, the adversary injects $(k+1) + \lfloor T/\varepsilon \rfloor+1$ packets into each link, i.e., $I=(k + \lfloor T/\varepsilon \rfloor + 2) \cdot (n^2-n)$ packets into the system.

On the other hand, since $G$ is a clique, any algorithm can successfully transmit at most $1$ packet per round in the entire network. Therefore, in $T = k \cdot \chi(H) = k \cdot (n^2-n)$ rounds at most $k \cdot (n^2-n)$ packets can be transmitted. So, at the end of a prefix of length $T$, there are at least $I - k \cdot (n^2-n) = (\lfloor T/\varepsilon \rfloor + 2) \cdot (n^2-n)$ packets remaining in the system. For $T$ approaching infinity, the number of packets remaining in the queues grows to infinity. This means that the queues are not bounded and the algorithm is not stable.
\end{proof}

\section{Reduction to the failure model}
\label{app:reduction}

First, let us explain the (wired) failure model~\cite{AlvarezBDSF05}. Given is a network graph $G$. A $(\rho,b)$-adversary in the failure model injects paths (packets) into $G$ and generates failures in such a way that in any interval $I$ the following inequality holds:
$$ Arr_e(I) + Fail_e(I) \leq \rho |I| + b,$$
where $Arr_e(I)$ is the number of packets injected during interval $I$ that pass through edge $e$ and $Fail_e(I)$ is the number of failures on edge $e$ generated during interval $I$. Each link $e$ that has some packets waiting in its queue can transmit a packet in every round, i.e., there are no collisions between edges.

There are known stable algorithms for packet scheduling in the failure model, such as NFS (Nearest-From-Source), SIS (Shortest-In-System), or FTG (Farthest-To-Go) against $(\rho, b)$-adversary with any $\rho < 1$~\cite{AlvarezBDSF05}.

\begin{lemma}
\label{lem:reduction}
Suppose we have a stable algorithm \textsc{Alg}  against any $(\rho'',b)$-adversary $ADV_{fail}$ in the failure model on graph $G$. Suppose we have a 
$(\rho', T)$-frequent schedule $\mathcal{S}$.

Then we can build a stable algorithm \textsc{$\mathcal{S}$-Plus-Alg} against any $(\rho,b)$-adversary $ADV_{RN}$ in the radio network model on graph $G$, for any $\rho$ such that $\rho < \rho'$ and $\rho'' \geq 1+\rho-\rho'$.
\end{lemma}
\begin{proof}
The stable algorithm in each round has two steps:
\begin{enumerate}
\item Determine which links transmit, according to a $(\rho', T)$-frequent schedule $\mathcal{S}$ for some parameters $\rho'$ and $T$,
\item Determine, for each link $e$, which packet awaiting in a queue of link $e$ to transmit, according to \textsc{Alg}.
\end{enumerate}

We can think of rounds when $\mathcal{S}$ does not successfully transmit a packet via link $e$ due to a collision as failures on link $e$ in the failure model. 
Schedule $\mathcal{S}$ guarantees that each link $e$ has at most $(1-\rho')T$ transmission blocked in any interval $I$ of length $T$. This means that each link $e$ has at most $Fail_e(I) \leq (1-\rho')T$ failures during $I$. Furthermore, $ADV_{RN}$ can inject at most $Arr_e(I) \leq \rho T + b$ packets passing through each edge $e$ during $I$.
$$ Arr_e(I) + Fail_e(I) \leq \rho T + b +  (1-\rho')T = T(1+\rho-\rho') + b$$
Therefore, the graph $G$ with packet arrivals from $ADV_{RN}$ and failures being collisions generated by $\mathcal{S}$ is an instance of the failure model with a
$(1+\rho -\rho', b)$-adversary. That means that using \textsc{Alg} to compute which packet to chose for each link at each round guarantees stability, provided $\rho'' \geq 1+\rho-\rho'$.
\end{proof}

\end{document}